\documentclass[runningheads]{llncs}
\usepackage[T1]{fontenc}
\usepackage{graphicx}
\usepackage{amssymb, amsmath, wasysym, mathrsfs}
\usepackage{xspace}
\usepackage[dvipsnames]{xcolor}
\usepackage{graphicx, wrapfig}
\usepackage{hyperref}
\usepackage[english]{babel}
\usepackage{tabularx}
\usepackage{booktabs, enumitem}
\usepackage{comment}
\usepackage{subcaption}
\usepackage{todonotes}

\usepackage{thmtools} 
\usepackage{thm-restate}

\newcommand{\NP}{\textsf{NP}\xspace}

\newcommand{\etal}{{et al.}\xspace}

\newcommand{\sublab}[1]{{{(#1)}}}

\newcommand{\settle}{{\text {\sc Settle}}\xspace}
\newcommand{\fullsettle}{{\text {\sc FullSettle}}\xspace}
\newcommand{\fix}{{\text {\sc Fix}}\xspace}

\newcommand{\desc}[1]{{\text {\tt #1}}}
\newcommand{\solu}[1]{{\text {\tt #1}}}

\def\nocomment{} 

\newcommand{\personaltodo}[4][noinline]{\ifdefined\nocomment\else\todo[#1,color=#2]{#3: #4}\fi}
\newcommand{\soeren}[2][noinline]{\personaltodo[#1]{WildStrawberry!30}{ST}{#2}}
\newcommand{\maarten}[2][noinline]{\personaltodo[#1]{blue!30}{ML}{#2}}
\newcommand{\guenter}[2][noinline]{\personaltodo[#1]{Orange!30}{GR}{#2}}
\newcommand{\alexandra}[2][noinline]{\personaltodo[#1]{teal!30}{AW}{#2}}

\newcommand{\change}{} 
\newcommand{\cutout}[1]{}

\newcommand{\UNP}{\textsc{1-SN}}

\spnewtheorem*{myproblem}{Unique Solution Nonogram}{\bfseries}{\itshape}

\declaretheorem[name=Observation]{observation}

\usepackage{environ}

\newcommand{\repeattheorem}[1]{%
  \begingroup
  \renewcommand{\thetheorem}{\ref{#1}}%
  \expandafter\expandafter\expandafter\theorem
  \csname reptheorem@#1\endcsname
  \endtheorem
  \endgroup
}

\NewEnviron{reptheorem}[1]{%
  \global\expandafter\xdef\csname reptheorem@#1\endcsname{%
    \unexpanded\expandafter{\BODY}%
  }%
  \expandafter\theorem\BODY\unskip\label{#1}\endtheorem
}

	\newcommand{\ldt}{\mathrel{.\,.}}
	\newcommand{\subproblem}[2]{P^{#1}_{#2}}
	\def\subproblem[#1,#2]{P^{#1}_{#2}}
	\def\subproblem[#1,#2]{\textit{Match}^{\,#1}_{#2}}

\usepackage{amsmath,amssymb}
\let\doendproof\endproof
\renewcommand\endproof{~\hfill\qed\doendproof}

\begin{document}
\title{On Solving Simple Curved Nonograms}
%
%
\author{
Maarten L{\"o}ffler\inst{1}\orcidID{0009-0001-9403-8856}
\and G{\"u}nter~Rote\inst{2}\orcidID{0000-0002-0351-5945}
\and Soeren Terziadis\inst{3}\orcidID{0000-0001-5161-3841}
\and Alexandra Weinberger\inst{4}\orcidID{0000-0001-8553-6661}
}

\authorrunning{M. L\"offler, 
G. Rote, S. Terziadis,
 and
A. Weinberger
}
%
\institute{
	Utrecht University, Utrecht, The Netherlands, \email{m.loffler@uu.nl} \and
	Freie Universität Berlin, Berlin, Germany, \email{rote@inf.fu-berlin.de} \and
	TU Eindhoven, Eindhoven, The Netherlands, \email{s.d.terziadis@tue.nl} \and
	FernUniversität in Hagen, Hagen, Germany, \email{alexandra.weinberger@fernuni-hagen.de}
}
\maketitle              

\soeren{Opinions on the title?}
\alexandra{I like it}
\begin{abstract}
Nonograms are a popular type of puzzle, where an arrangement
of curves in the plane (in the classic version,  a rectangular grid) is given together with a series
of hints, 
indicating which cells of the subdivision are to be
colored. 
The 
colored cells yield an image.
Curved nonograms use a curve arrangement rather than a grid, leading to a closer approximation of an arbitrary solution image.
While there is a considerable amount of previous work on the natural
question of the hardness of solving a classic nonogram,
research on curved nonograms has so far focused on their creation,
which is already highly non-trivial.
We address this gap by providing algorithmic and hardness results for curved nonograms of varying complexity. 
\end{abstract}


\section{Introduction}
\label{sec:Introduction}

Nonograms, also known as 
{\em Japanese puzzles}, {\em paint-by-numbers}, or {\em griddlers},
are a popular puzzle type where one is given an empty grid in which some
 grid cells are to be colored (filled);
the remaining cells remain empty (unfilled).
For every row and column,
there is a \emph{description} that
constrains the set of colored grid cells in this row or column.
The description specifies how many consecutive blocks of cells should be filled and how large these blocks are.
Two filled blocks need to be separated by one or more unfilled cells.
A solved nonogram typically results in a picture (see Figure~\ref {fig:nonograms}).

Nonograms provide an accessible and contained environment for logical deduction.
They have been used successfully to teach logical thinking~\cite {10.1145/1141904.1141906,10.1007/978-3-319-94619-1_38}, and have been shown to stimulate brain activity to prevent dementia~\cite {kasinathan2020logical}.

Batenburg~\etal~\cite{Batenburg09} introduce the notion of a {\em simple} nonogram, which can be solved efficiently.
A nonogram is {\em simple} when it can be solved by only looking at a single row or column at a time.
More precisely, they consider a nonogram simple if it can be solved by repeatedly considering a row or column, enumerating all possible solutions for it
that are consistent with the fixed cells determined so far,
and fixing all cells which have the same value in any possible solution.
This procedure is called \emph{settling a row/column} (or simply \settle) and will be considered in more detail in the preliminaries (Section~\ref{sec:prelims}).
Note that since repeated application of settling a row or column is a deterministic process,
the existence of multiple solutions for a nonogram immediately implies that it cannot be simple; however, the converse is not true, i.e., there are uniquely solvable nonograms
that are not simple.
In fact, Batenburg and Kosters introduce a 
whole
hierarchy of complexity for nonograms, depending on the number of rows and columns which have to be considered simultaneously (by a specific solver)
in order to definitively identify a cell whose status
can be settled;
puzzles with unique solutions can be found at all levels of this hierarchy.


Nonogram puzzles that appear in newspapers or similar platforms tend to be of this simple type~\cite{BATENBURG20091672},
which contradicts to some extent the popular opinion that all interesting games and puzzles are \textsf{NP}-hard~\cite{cormode2004hardness,viglietta2014gaming}.


\begin{figure}[tbp]
	\centering
	\subfloat[\label{fig:nonograms_a}]{%
		\includegraphics[width=.3\linewidth,page=1]{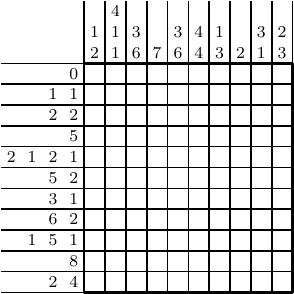}
	}\hfil
	\subfloat[\label{fig:nonograms_b}]{%
		\includegraphics[width=.3\linewidth,page=2]{figures/nonograms}
	}\hfil
	\subfloat[\label{fig:nonograms_c}]{%
		\includegraphics[width=.3\linewidth,page=3]{figures/nonograms}
	}
	\caption{
		\sublab{a} A classic nonogram puzzle. 
		\sublab{b} An inference based on the highlighted clue.
        \sublab{c} The solved nonogram.}
	\label{fig:nonograms}
\end{figure}  

\subsection {Solving Nonograms}

Beside Jan Wolter's online survey~\cite{webpbnSurveyPaintbyNumber}, there has been substantial academic interest in nonograms.
The natural question is to ask whether a given nonogram can be solved or to come up with an algorithmic approach to compute such a solution.
A number of solvers using various strategies have been presented in the literature.
These include heuristic approaches~\cite{salcedo2007solving},
DFS-based solving
methods~\cite{5212614,DBLP:journals/apin/YuLC11,stefani2012solving},
genetic
algorithms~\cite{DBLP:journals/apin/Tsai12,7738765,7521597,alkhraisat2016dynamic,DBLP:journals/icga/ChenL19},
line-by-line solving combined with probing (using low probability
guesses to quickly achieve
contradictions)~\cite{DBLP:journals/dam/BerendPRR14}, \textsf{SAT}
solvers~\cite{DBLP:journals/jair/AMCS13},
integer linear programming~\cite{khan2020solving}, 
and a combination of heuristics and neural networks~\cite{BUADESRUBIO2024100652}.
The performance of two general solving strategies (DFS and so called
soft computing)
has been experimentally 
compared
on a small set of four nonogram instances~\cite{WIECKOWSKI20211885}.

The computational problem of deciding if a nonogram has a solution is \textsf{NP}-complete, as was
first shown by Uada and Nagao~\cite {ueda96}; see also \cite
{hoogeboom14,rijn12}.
Note that this of course implies that computing this solution is also at least \textsf{NP}-hard.
Uada and Nagao additionally prove that, given a nonogram and a solution, deciding if this solution is unique is also \textsf{NP}-hard, via a parsimonious reduction from three-dimensional matching.

In contrast,
Batenburg and Kosters~\cite{BATENBURG20091672} gave a polynomial-time algorithm that, given a sequence of partially settled cells and a corresponding description, 
finds a cell that is either filled or unfilled in every possible solution,
if such a cell exists.
Their procedure can be used to either solve a given nonogram in polynomial time or decide that it is not simple.




\subsection {Curved Nonograms}
\begin {figure}[tbp]
\subfloat[]{%
	\includegraphics[page=1]{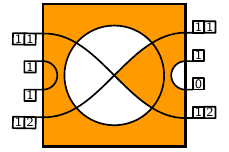}
}\hfil
\subfloat[]{%
	\includegraphics[page=2]{figures/difficulty}
}\hfil
\subfloat[]{%
	\includegraphics[page=3]{figures/difficulty}
}
\caption 
{
	Three types of curved nonograms of increasing
	complexity~\cite
	{kjplvk-dagjpp-19}, shown with solutions.
	\sublab{a} \emph {Basic} puzzles have no popular faces.
	\sublab{b} \emph {Advanced} puzzles may have popular faces, but no self-intersections.
	\sublab{c} \emph {Expert} puzzles have self-intersecting curves.
	%
}
\label {fig:difficulty}
\end {figure}

Van de Kerkhof~\etal~\cite {kjplvk-dagjpp-19} introduced {\em curved} nonograms, a variant in
which the puzzle is no longer played on a grid but on any arrangement
of 
curves (an example is shown in Figure~\ref {fig:difficulty}); see also~\cite {DBLP:conf/cccg/Kreveld18}.
For distinction, we will refer to the nonograms played on a grid as \emph{classic nonograms}. 
In curved nonograms, the numbers of filled
faces of the arrangement in the sequence of faces that
appear \change{along a} side \change{of a} curve\change{, are specified by a description (one on each side)}.
Curved nonograms 
allow
cells with more organic shapes than classic nonograms, and thus lead to clearer or more specific
pictures.
Van de Kerkhof~\etal focus on heuristics to automatically generate such puzzles from a desired solution picture by extending curve segments to a complete curve arrangement.

Additionally, they 
define three
different levels
of complexity
of  curved nonograms
--- not in terms of how hard it is to \emph{solve} a puzzle, but how hard it is to understand the rules (see Figure~\ref {fig:difficulty}).
It turns out that these difficulty levels nicely correspond with properties of the underlying curve arrangement as observed by De Nooijer~\etal~\cite
{
  DBLP:journals/jgaa/NooijerTWMMLR24}
  (see \cite{DBLP:conf/gd/NooijerTWMMLR23} for the conference version).
Specifically, \emph{basic} curved nonograms are \change{exactly} the puzzles in which each description corresponds to a sequence of {\em distinct} faces.
The analogy with descriptions in classic nonograms is straightforward. 
In an \emph {advanced} curved nonogram, a face may be incident to the same curve multiple times, but only on the same side,
and therefore a face can appear more than once in a sequence.
If such a face is filled it is also counted multiple times when checking consistency with a description; in particular, it is no longer true that the sum of the numbers in a description is equal to the total number of filled faces incident to the curve.
\emph {Expert} curved nonograms may have descriptions in which a single face is incident to the same curve on \emph {both} sides (which corresponds to the presence of a self-intersecting curve in the arrangement).

Research on curved nonograms has so far focused on their production.
Klute, L\"offler and N\"ollenburg~\cite{KLUTE2021101791} investigate the geometric problem of adding descriptions to the ends of curves and provide 
polynomial-time algorithms for restricted cases and hardness results for the general problem. 
De Noojier~\etal~\cite {
  DBLP:journals/jgaa/NooijerTWMMLR24} aim to eliminate
all faces with multiple incidences to the same curve
 (so-called \emph{popular faces}) 
 from a nonogram by adding one additional curve to the arrangement.
The same goal was recently pursued in a Dagstuhl seminar working group~\cite{buchin_et_al:DagRep.12.2.17}, which aimed to do so by reconfiguring the curve arrangement through local crossing resolution.


\subsection {Contribution}

In this paper, we investigate for the first time the computational problem of {\em solving} curved nonograms.
In particular, we investigate how the concept of {\em simple} nonograms translates to curved nonograms.
After some preliminaries in Section~\ref{sec:prelims}, we present in Section~\ref{sec:advanced} a dynamic program which leverages the nested structure of popular faces in advanced nonograms to check for a given sequence of faces along a curve, some of which are already filled or unfilled, if it can still be extended to a solution that is consistent with a given description.
This implies a procedure solving simple advanced curved nonograms in \change{$O(l^7)$ time, where $l$ is the length of the longest description.
This runtime can be improved to $O(l^6)$ by using an additional top-down phase of the dynamic program.}
In the case the nonogram is basic, the dynamic program coincides with a special case of the one presented by Batenburg and Kosters~\cite{BATENBURG20091672}, showing that simple basic curved nonograms can be solved in the same way as simple classic nonograms.
Then Section~\ref{sec:expert} shows that self-intersecting curves likely make curved nonograms significantly harder to solve, since even simple curved expert nonograms are at least as hard to solve as 
classic nonograms with a guaranteed unique solution.
We close with some further research questions in Section~\ref{sec:conclusion}.


\section {Preliminaries}\label{sec:prelims}

In this section we introduce the basic concepts and notation as well as the basic problems, which naturally arise in the context of solving nonograms.

\subsection{Nonograms}

Let $\cal A$ be a curve arrangement consisting of $h$ curves $A_1, \ldots, A_h$ all contained in and starting and ending at a rectangle called the \emph{frame}.
Every 
piece of a curve $A$ between two consecutive intersections (or the start or end of $A$) is a curve segment of~$A$.
A face of $\cal A$ (also called cell) is \emph{popular} if two or more curve segments incident to the face belong to the same curve.
Every cell initially has the value \emph{\change{unsettled}}
which we 
denote with $\desc{?}$.
If a value is assigned one of the two values \emph{empty} ($\desc{0}$) or \emph{filled} ($\desc{1}$) we say that the cell is \emph{settled}.
Here we follow the notation of Batenburg and Kosters~\cite {BATENBURG20091672}.

We choose an arbitrarily orientation for each curve;
\change{accordingly,}
a face $f$ 
\change{incident} 
to a curve segment $s$ is said to be \emph{on the left} or \emph{on the right} side of~$s$.
Let $s_1, s_2, \ldots, s_k$ be the curve segments of a curve $\ell$ in $\mathcal A$.
We call the list of faces $f_1, \ldots, f_k$, s.t. $f_i$ is on the right (left) of $s_i$ the \emph{right (left) sequence} $S^r_\ell$ ($S^l_\ell$) of $\ell$.
Popular faces 
can appear multiple times in the same sequence,
and if 
\change{$\ell$ is a} 
self-intersecting curve, faces can appear in both sequences.

\change{A specification for a sequence $S$ is a string}
$\Psi^S = \psi_1\change{\psi_2} \ldots \psi_k \in 
\change{\{\desc{0},\desc{1},\desc{?}\}}
^k$, \change{and it encodes the current state of knowledge about the faces in the sequence.} 
If $\Psi^S$ contains no $\desc{?}$, then it is a \emph{fix}. 
If the sequence in question is clear from context, we may omit the superscript.
If for two specifications $\Psi$ and $\Psi'$
\change{of the same sequence~$S$}  
%
it holds that either $\psi_i = \desc{?}$ or $\psi_i = \psi'_i$
\change{for all~$i$}, we say that $\Psi'$ \emph{refines} $\Psi$.

A description $D = d_1, \ldots, d_t$ 
is a list of $t$ numbers.
One such number $d_i$ will 
be called a \emph{clue} of $D$.
A fix $\Psi$ of $S$ is \emph{consistent with $D$} if and only if $\Psi$ contains exactly $t$ maximal blocks of consecutive $\desc{1}$s and the $i$-th 
block consists of exactly $d_i$ $\desc{1}$s. 
Since these blocks are maximal,
\change{consecutive}
blocks 
are separated by one or more $\desc{0}$s.
A 
specification $\Psi$ of $S$ is \emph{consistent with $D$} if 
there exists a fix that is consistent with $D$ and refines $\Psi$.

\change{
In a curved nonogram, a face can appear more than once 
along a curve.
This leads to additional constraints in the form of equations $\psi_i=\psi_j$.
We encode this by a sequence of letters $f_1\ldots f_l$, like                
$\texttt{ab}\texttt{cd}\texttt{efd}\texttt{b}\texttt{gb}\texttt{h}$,
where repeated letters indicate positions that belong to 
the same face.
For example, the 2nd, 8th, and 10th edges lie on a common face, marked~\texttt{b}.
We call this the \emph{letter description} of the sequence.} 

\change
{We will only consider specifications $\Psi
$ 
that fulfill all equality constraints.
%
}

A \emph{curved nonogram} $C$ consists of a curve arrangement together with a set of descriptions and specifications (one for each sequence in $C$ respectively).
If all specifications are fixes, we say the nonogram is solved and conversely \emph{solving} a given nonogram means obtaining a fix for every specification that is consistent with its description.
For any $i\leq j$ we write $i\ldt j$ for the list of numbers between $i$ and $j$ (including both).

\subsection{Settling and Nonogram Complexity}
Given a specification $\Psi$, which is consistent with a description $D$, obtaining a specification $\Psi'$ that refines $\Psi$ and is still consistent with $D$ is called \emph{making progress on $\Psi$}.
The procedure $\settle(\Psi, D)$ takes a specification and a description and (if possible) returns a specification $\Psi'$, which refines $\Psi$ and is consistent with $D$.
It does so by settling any \change{unsettled} cells to be filled (or empty) if they have the same value in all possible fixes which refine $\Psi$ and are consistent with~$D$.
This procedure is illustrated in Figure~\ref{fig:settle}.

\begin{figure}[tbp]
    \centering
    \includegraphics{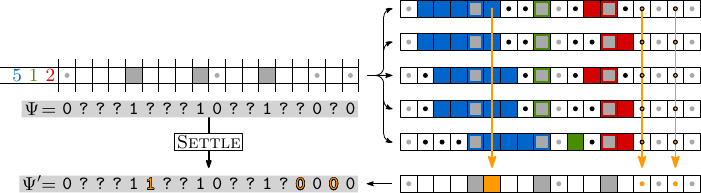}
    \caption{
    An example of the \settle function.
    Filled cells are shown with a filled square, empty cells with a dot all other cells are \change{unsettled}.
    \change{For example, in} a row with specification $\Psi$ and description $D =$ 5-1-2 (shown in the top left) there are five possible fixes of $\Psi$ consistent with $D$ shown in the top right (cells already settled in $\Psi$ are colored gray, newly settled ones black).
    The blocks corresponding to the three parts of $D$ are indicated in corresponding colors.
    One cell is filled and two are empty over all five fixes, which is indicated with yellow entries.
    }
    \label{fig:settle}
\end{figure}

Note that there can be exponentially many such fixes.
However, while \settle is defined via equality of the value of a cell over all possible fixes, an implementation of \settle does not necessarily need to enumerate all possible fixes to find such a cell.
For example in a classic nonogram, the dynamic program of Batenburg and Kosters~\cite{BATENBURG20091672} finds such a cell in polynomial time or decides that no such cell exists.
Applying \settle to all rows and columns of a nonogram until no progress can be made is called a \fullsettle.

If every specification of a nonogram has a fix consistent with its description it is \emph{solvable} and correspondingly we will call a nonogram in which every specification is a fix the \emph{solution} of the nonogram.
If a nonogram is solvable via a \fullsettle it is called \emph{simple}.
We remark that this definition is in line with~\cite{BATENBURG20091672}, whose algorithm can solve simple classic nonograms in polynomial time.

	\section {Solving Simple Advanced Curved Nonograms}
    \label{sec:advanced}

    In this section we present a dynamic program which, given a sequence $S$ together with a specification $\Psi$ and a description $D$ 
    \change{decides} 
    in polynomial time if there exists a fix $\Psi'$ consistent with $D$ that refines $\Psi$.
    This is \change{analogous} to the existing dynamic program by Batenburg and Kosters for classic nonogram~\cite{BATENBURG20091672}. readers familiar with their work will easily spot the parallels; however, the presence of popular faces requires the maintenance of an additional data structure.
    The application of our algorithm to simple basic curved nonograms (i.e., those without popular faces) is discussed at the end of the section.

\change{
The  property of advanced nonograms that is crucial for us is that the equality constraints are properly nested:
\begin{observation}
\label{obs:nested}
    Let $i, j, k$ and $l$ be four indices of letters in the letter description of a sequence $S$ belonging to a curve $A$ in an advanced curved nonogram, such that $f_i = f_j \ne 
    f_k = f_l $. 
    W.l.o.g. assume $\min(i,j,k,l) = i$ and $k<l$.
    Then either (i) $j<k \land j<l$ or (ii) $j>k \land j>l$.
\end{observation}
\begin{proof}
   We only have to exclude the order 
   $i<k<j<l$.
Assume w.l.o.g.\ that we are considering the left side of~$A$. 
    Let $a$ and $b$ be points on the $i$-th and $j$-th segment of $A$,
    respectively.
    Since these segments lie on a common face, we can connect
    $a$ and $b$ by a curve $B$ in that face, on the left side of $A$.
    Let $A[a,b]$ be the subcurve of $A$ between $a$ and~
    $b$.
    Then $C = A[a,b] \cup B$ is a Jordan curve (simple and closed).
    If $i<k<j<l$, 
    $C$ encloses the face on the left side of the $k$-th segment, but
    it does not enclose the face on the left side of the $l$-th segment.
Hence, these faces cannot be identical,
contrary to our assumption $f_k=f_l$,
and therefore, 
the order $i<k<j<l$ is impossible.
\end{proof}
We mention that the nesting property of \autoref{obs:nested} is the only property on which our algorithm relies. If a curve $A$ has self-intersections but there are no faces that lead to a violation of the nesting property, our algorithm can be applied.
}

 \guenter{the following (original) stmt. is now wrong. Now we only check consistency?}\soeren{Indeed, I thought I had changed that but apparently no. Changed into what we are doing. Progress is covered later}
	\begin{theorem}\label{thm:advanced}
        Consistency of a 
        sequence $S$ of length $l$
        with
        a description $D$ with $\sum_{d\in D} d = k$ can be decided in time $O(k^3l)=O(l^4)$.
	\end{theorem}
	
	\begin{proof}
		In a bottom-up phase of the dynamic program we try to match larger and larger intervals of the specification with larger and larger parts of the description.
        In a subsequent top-down phase we will discover which assignments are consistent with an overall solution.


		We translate $D = d_1d_2\ldots d_t$ to a list $a_1 a_2 \ldots a_k$ of \desc{0}s and \desc{1}, by creating blocks of $d_i$ \desc{1}s for every $1 \leq i\leq t$ and concatenating them with one \desc{0} between consecutive blocks.
        Additionally we artificially pad the list by an extra \desc{0} at the beginning and at the end.
        Note that this assumption implies that every row and column starts and ends with an empty cell.
        Every nonogram can obviously be padded with empty rows and columns to achieve this.
		For example, \desc{5-1-2} is translated to \desc{011111010110}, with the understanding that a \desc{0} has the potential to stretch to an arbitrary larger number of unfilled cells.
        A specification is now consistent with this \desc0-\desc1-string variant of a description if one
        can create two equal strings by
replacing every \desc{?} in the specification with either \desc{0} or \desc{1} and replacing any \desc{0} in the description with one or more \desc{0}.
        The following example shows this for $D =$ 5-1-2.

                
		\smallskip
                
\begin{tabular}[l]{ll@{ }l}
  \desc0-\desc1-string of $D$:&
                $a_1 a_2 \ldots a_i  \ldots a_k $&$=\desc{011111010110}$\\[1ex]
  Specification $\Psi$:&
              $      \psi_1 \psi_2 \psi_3\ldots \psi_j  \ldots \psi_{l-1}\psi_l $&$=\solu{0???1???10??1??0?0}$\\
\end{tabular}
\medskip
		

		In the finished nonogram, all \desc{?}'s should be turned into \desc{0}'s or \desc{1}'s,
		subject to the requirement
		that the resulting sequence fits the specification.
		We want to know whether a particular \desc{?} 
        can be turned only into \desc{0} or only into \desc{1} in \emph{all} possible solutions,
		because then this \desc{?} can be fixed to this value.
        In other words we aim to implement the \settle procedure.

        \change{Recall that}
		the finished solution must satisfy certain equations $\psi_i=\psi_j$ when two edges are incident
		to a common face, \change{which we have}
encoded by a
\change{letter description}
$f_1\ldots f_l$, like                
$\texttt{ab}\texttt{cd}\texttt{efd}\texttt{b}\texttt{gb}\texttt{h}$,
where repeated letters indicate that edges belong to the same
face.
%
\change{According to 
\autoref{obs:nested}},
these repeated occurrences are
\emph{nested}:
The pattern
$\ldots x \ldots y \ldots x \ldots y \ldots$ cannot occur in
 the sequence. 

\begin{figure}[tbp]
	\centering
	\subfloat[\label{fig:decomp_arrangement}]{%
		\raise 8mm\hbox
        {\includegraphics[scale=.9,page=6]{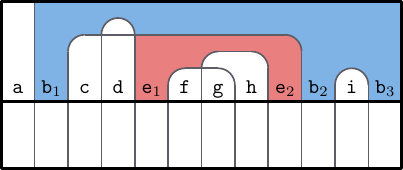}}
	}\hfil
	\subfloat[\label{fig:decomp}]{%
		\includegraphics[scale=.9,page=7]{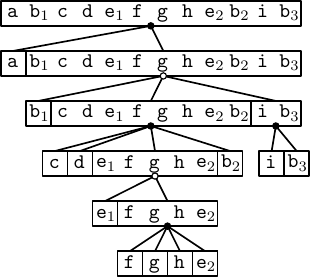}
	}
	\caption{(a) A schematic representation of a curve arrangement indicating the face incidences for top side of the horizontal line and (b) its hierarchical decomposition into
   subintervals.
   For clarity,
  multiple occurrences of the same face (such as
  $\texttt{b}_1,\texttt{b}_2,\texttt{b}_3$) are
  distinguished by 
  indices.                 
  A white node denotes a decomposition
  of a complete group into brackets; a black node
  denotes decomposition of a bracket into complete groups.
 Black and white nodes occur in alternate levels of the tree.
  }
	\label{fig:decomp_all}
\end{figure}

\begin{figure}[htb]
  \centering
   \includegraphics[scale=.9,page=8]{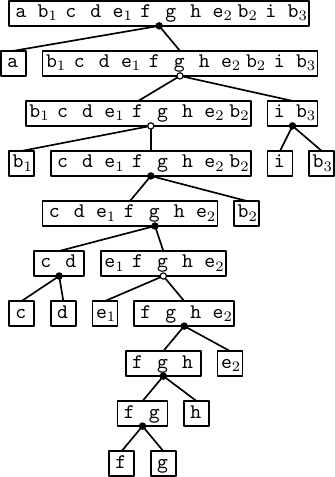}
    \caption{The  binary composition tree $\mathcal{T}$ corresponding to the tree of \autoref{fig:decomp}, 
    in which nodes
      of
    higher degree have been replaced by sequences of binary
    nodes. This is a binary tree whose leaves are the singleton intervals.}
  \label{fig:decomp-binary}
\end{figure}
		
		We solve the following subproblems $\subproblem[i\ldt i',j\ldt j']$, for every
		$1\le i\le i'\le k$ and for a certain set $J$ of
        $2l-1$
        selected intervals $j\ldt j'$ with $1\le j\le j'\le l$:
		\begin{quote}
			Can the \desc{?}'s in $\psi_j\ldt \psi_{j'}$ be turned into \desc0's or \desc1's such that
			the resulting string 
            \change{is consistent with}
            the description $a_i\ldt a_{i'}$?
		\end{quote}
The subproblem
$\subproblem[i\ldt i',j\ldt j']$ results in a Boolean value 
\emph{true} or \emph{false}. Accordingly, we will say that a subproblem
is consistent or inconsistent.
The set $J$ of intervals $j \ldt j'$ of the curve that we consider is defined as
follows.
		
		Suppose that the cells $j_1<j_2<\dots<j_m$
        are
the cells belonging to some common face:
$f_{j_1}=f_{j_2}=\dots=f_{j_m}$.
		We call the interval
		$j_1\ldt j_m$ a \emph{complete group}, and we call the intervals
		$j_1\ldt j_p$, for $p=1,\ldots,m$ the \emph{progressive sleuths}.
		The first progressive sleuth is the singleton
		interval $j_1\ldt j_1$.
		If a face occurs only once along the curve, at position $j$, then the
singleton interval $j\ldt j$ forms a complete group.
		
		An interval $j_p+1\ldt j_{p+1}$
		between two successive occurrences of
		the same face, including the second occurrence but not the first,
		is called a
		\emph{bracket}.
Such a bracket consists of a
(possibly empty) succession of complete groups,
followed by an occurrence of the face of the enclosing group at position $j_{p+1}$.
		%
		%
		We consider also the whole interval $1\ldt l$ as a bracket although it lacks
		the final element of the enclosing group (see Figure~\ref{fig:decomp_all} for an illustration).
		
		We will build up the whole curve $1\ldt l$, starting from singleton
		intervals
		$j\ldt j$.
        These can be seen as the leaves of a binary composition tree $\mathcal{T}$ (see Figure~\ref{fig:decomp-binary})
    This binary tree represents how we {will }{}combine certain pairs of consecutive intervals
		$j\ldt j'$
		and $j'+1\ldt j''$ into larger intervals
		$j\ldt j''$. This is done as follows.
		
		Every complete group is built up from left to right by successive
		addition of brackets:
		\begin{equation*}
			[j_1\ldt j_p] \cup
			[j_p+1\ldt j_{p+1}] =
			[j_1\ldt j_{p+1}]
		\end{equation*}
		Similarly, every bracket is built up from left to right by successive
		addition of the complete groups that make it up (plus the final cell of
		the enclosing group).
        In total, a set $J$ of $2l-1=O(l)$ intervals $j\ldt j''$ are considered. Each such
		interval with $j<j''$ -- an internal node in $\mathcal{T}$ -- is built in a unique way from two disjoint
		subintervals in $J$:
		\begin{equation}
			\label{eq:subintervals}
			[j\ldt j''] =
			[j\ldt j'] \cup
			[j'+1\ldt j'']
		\end{equation}
See Figure~\ref{fig:decomp-binary} for an example.
		
		The singleton subproblems of the form
		$\subproblem[,j\ldt j]$ are trivial to solve:
        For a description of length 1, we have $\subproblem[i\ldt i, j\ldt j] \iff \psi_j=\desc{?} \lor \psi_j=a_i$, while $\subproblem[i\ldt i', j\ldt j]$ is trivially inconsistent for $i<i'$, since a single cell can never be consistent with a description of length larger than 1.
		
Following the decomposition~\eqref{eq:subintervals},
 each larger subproblem of the type
 $\subproblem[,
 j\ldt j'']$ with $j''>j$
is associated to two families of smaller subproblems 
$\subproblem[,
j\ldt j']$ and
$\subproblem[,
j'+1\ldt j'']$,
for some fixed $j'$.
We solve these subproblems by the following recursion:
\begin{equation}
  \label{eq:recursion}
  \begin{aligned}
				\subproblem[i\ldt i'', j\ldt j''] {\iff} 
				&
				\bigvee_{i': i\le i' \le i''-1}\left(
				\subproblem[i\ldt i', j\ldt j']
				\land 
				\subproblem[i'+1\ldt i'', j'+1\ldt j'']
				\land a_{i'}=a_{i''}\right)
				\\&
				\lor
                \bigvee_{i': i\le i' \le i''}\left(
	a_{i'}=0
				\land 
				\subproblem[i\ldt i', j\ldt j']
				\land 
				\subproblem[i'\ldt i'', j'+1\ldt j'']
				\land a_{i'}=a_{i''}\right),
			\end{aligned}
\end{equation}
	where
		the final condition  $a_{i'}=a_{i''}$ is present only in case of
		composing a progressive sleuth with a bracket.
		It ensures that occurrences of the same face have the same color \desc{0} or \desc{1}; 
        when combining two complete groups, there are no shared faces that need to be considered, and the
        condition  $a_{i'}=a_{i''}$ is omitted. 
        
	The first clause considers all possibilities of
        splitting the interval
        $i\ldt i''$ into two disjoint parts
$i\ldt i'$ and $i'+1\ldt i''$. 
		The second clause considers in addition the possibility that the two parts of the curve can use
		overlapping parts of the description if the overlap is a~\desc{0}
		
This completes the description of the bottom-up phase.
The \emph{target problem} $\subproblem[1\ldt k, 1\ldt l]$
describes the original problem: consistency of
the whole specification $\Psi$ with the complete description $D$.        
		
		In total, there are $O(k^2l)$ subproblems, and each subproblem can be
		evaluated by trying $O(k)$ choices for $i'$, for a total running time
of $O(k^3l)$.
\end{proof}

\subsection{Making progress}
Having solved the consistency problem, we immediately get a polynomial-time
solution algorithm for making progress.

\begin{theorem}\label{thm:naive_progress}
    Given a single sequence $S$ of length $l$ and a description $D$ with $\sum_{d\in D} d = k$ we can make progress or decide that no progress can be made in time $O(k^3l^2)=O(l^5)$.
\end{theorem}
\begin{proof}
    We tentatively set a \desc{?} letter to \desc{0} or \desc1 and check
    consistency again. If one of the options is inconsistent, then we know
    that \desc?
    must be replaced by the other letter, thus making progress.
    This is repeated for each of the at most $l$ occurrences of \desc{?}.
\end{proof}

However, we can solve this more efficiently and avoid the additional factor $l$ by a top-down phase,
in which we
mark certain subproblems as
		\emph{extensible}.

\begin{theorem}\label{thm:app:top_down_progress}
    Given a single sequence $S$ of length $l$ and a description $D$ that describes $k$ ones we can make progress or decide that no progress can be made in time $O(k^3l)=O(l^4)$.
\end{theorem}
\begin{proof}
We call a subproblem
		$\subproblem[i\ldt i',j\ldt j']$ is \emph{extensible} if it is
		consistent 
		and in
		addition,
		some solution that fits the description $a_i\ldt a_{i'}$ can be extended to
		a complete solution by setting the remaining \desc{?}'s outside the
		substring $b_j\ldt b_{j'}$ appropriately.
		
		We begin by marking the target problem
		$\subproblem[1\ldt k, 1\ldt l]$,
		as extensible, assuming it is consistent. 
		Then we use the recursion~\eqref{eq:recursion} in reverse.
		If $\subproblem[i\ldt i'', j\ldt j'']$ is extensible, then,
		if any of the parenthesized clauses on the right-hand
                side
of~\eqref{eq:recursion}
                holds for
		some $i'$, we mark the two
corresponding subproblems $\subproblem[i\ldt i', j\ldt j']$ and
		$\subproblem[i'(+1)\ldt i'', j'+1\ldt j'']$ as extensible.

		Finally we look at each \change{unsettled} position $j$ with $b_j=\desc{?}$, and we
		check for which description positions $i$
		the problem $\subproblem[i\ldt i, j\ldt j]$ is extensible.
		If all extensible problems among these have $a_i=\desc{0}$, we can conclude that
		$b_j$ must be set to \desc{0}, and settle an \change{unsettled} color
		in this way.
		Similarly, if all extensible problems have $a_i=\desc1$,
		we can fix the \change{unsettled} value $b_j$ to \desc1 at this position.
\end{proof}

		



    Theorem~\ref{thm:naive_progress} (or Theorem~\ref{thm:app:top_down_progress}) can be used to obtain the following corollary by the simple fact that there are only a linear number of rows and columns and cells in the nonogram.
    After applying the dynamic program once to every row and column, we must have made progress on at least one sequence, so there is at most an overhead of $O(l^2)$.
    
    \begin{corollary}
	    Simple advanced curved nonograms can na\"ively be solved in time~$O(l^7)$, or in time $O(l^6)$ using the top-down phase of Theorem~\ref{thm:app:top_down_progress}.
    \end{corollary}
    
	\subsection {Back to Basic Nonograms}

When our algorithm is applied to a curve in a basic nonogram,
there are no groups, and the whole sequence is just one
bracket whose sequence of  ``complete groups'' consists of
	singletons.
The decomposition tree degenerates, and    
the algorithm simply grows the intervals $1\ldt i$ and  $1\ldt j$
by adding one symbol at at time.
Here our dynamic program reduces to the seminal 
 algorithm of Batenburg and Kosters~\cite{BATENBURG20091672}
(which is actually more general because it can deal with 
    a specified \emph{range} of lengths for each \desc{1}-block instead of a fixed length).
	

	\section {Solving Simple Expert Curved Nonograms}
    \label{sec:expert}

    In this section we show that solving a simple curved expert nonogram is at least as hard as finding the solution to a not necessarily simple classic nonogram provided that the classic nonogram has a unique solution.
    \change{\begin{myproblem}[\UNP]\label{prob:nonogram_unique}
        Given a classic nonogram $N$ \change{with} the guarantee it has a unique solution, find the solution for the nonogram $N$.
    \end{myproblem}}

    Note that  testing if a classic nonogram has a solution is \textsf{NP}-hard in general~\cite{ueda96}.
    This of course implies that finding such a solution is also \textsf{NP}-hard.
    Ueda and Nagao~\cite{ueda96} also show that testing if a given nonogram has more than one solution, even if we are given a solution, is 
    \textsf{NP}-hard.

    \begin {figure}[tbp]
  \centering \includegraphics {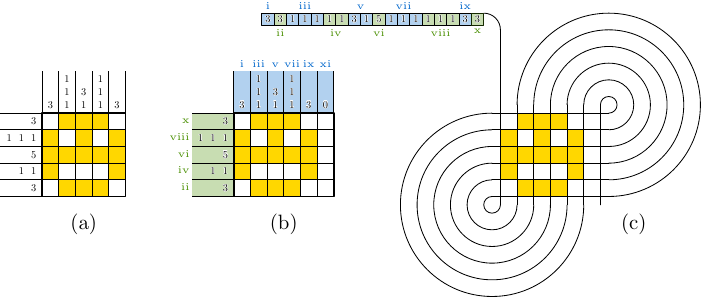}
  \caption
  { (a) A classic nonogram with $w=5$ columns and $h=5$ rows.  (b) \change{In the reduction we pad a given nonogram to guarantee it has one more column than rows. The result here is a}
    padded nonogram with $w=6$ columns and $h=5$ rows. \change{The hints are annotated with the order in which they are collected into} (c) a single
    self-intersecting {\em vital} curve that contains all grid lines
    of the classic nonogram \change{and all descriptions in one \emph{vital} description}. 
    }
  \label{fig:hard1} 
\end {figure}

    However this does not directly imply that finding a solution for a classic nonogram is still \textsf{NP}-hard if we are guaranteed that it has a unique solution.
    By providing a reduction from \UNP{} to \textsc{Simple Curved Expert Nonogram} we still show that finding the solution to a simple curved expert nonogram is at least as hard as solving \UNP{}.

\subsection {High Level Overview}


The high-level overview of our reduction is as follows.
    We first describe the construction of a curved expert nonogram $C$ based on a given not-necessarily-simple classic nonogram $N$.
    If we are guaranteed that $N$ has a unique solution, then $C$ will equally have only a single unique solution and additionally we will show that $C$ is simple.
    Next we argue that the sequences along all but one curve can be trivially filled (using the \settle procedure) by simply filling all sequences whose description require the entire sequence to be colored black.
    This will yield a partially filled curved expert nonogram, in which all cells, which are not yet colored in, are part of the right sequence $S^r_\ell$ along a single curve $\ell$.
    Moreover the specification $\Psi$ of this sequence includes already filled chains of cells and there is a one-to-one correspondence between any chain of unsettled cells to a row or column of the input classic nonogram.
    Given any fix which refines $\Psi$ and is consistent with the description of $S^r_\ell$ will fill in all remaining cells of $C$ and immediately yields a solution for $N$.

    Since $C$ is simple, it can be solved with an application of $\fullsettle$.
    Therefore a polynomial time algorithm for $\fullsettle$ on curved expert nonograms would imply that classic nonograms can be solved in polynomial time, if we are guaranteed that their solution is unique.

    \subsection{Constructing the curved expert nonogram}\label{sec:hardness_construction}
    Consider a classic nonogram $N$ with $w$ columns and $h$ rows, as in Figure~\ref {fig:hard1}(a).
    We will assume w.l.o.g. that  we have $w = h+1$; this can be achieved through appropriate padding, see Figure~\ref {fig:hard1}(b).
	Note that adding empty (or completely filled) rows or columns does not change the difficulty of the puzzle.
    All cells of $N$ will also be contained in the curved expert nonogram $C$ we created based on~$N$.
    We will call these cells the \emph{original cells}.
    
	Now, conceptually, we will trace a single {\em vital} curve through all $w+1$ vertical and $h+1$ horizontal line segments that make up the grid of the puzzle (excluding the section that contains the descriptions); refer to Figure~\ref {fig:hard1}(c).
	Doing this will concatenate the descriptions from all rows and columns of the original nonogram into a single description; specifically, it will intersperse the descriptions of the columns (from left to right) and the rows (from top to bottom). We will refer to the resulting description as the {\em vital} description.

However, this alters the difficulty of the puzzle, as the information which sections of the vital description belong to separate rows or columns is lost.
		To solve this, we again pad the original nonogram, but now with rows and columns, which we will force to be entirely colored in in a solution of $C$ as follows.

	\begin {figure}[tbp]
		\centering
		\includegraphics[scale=.7]{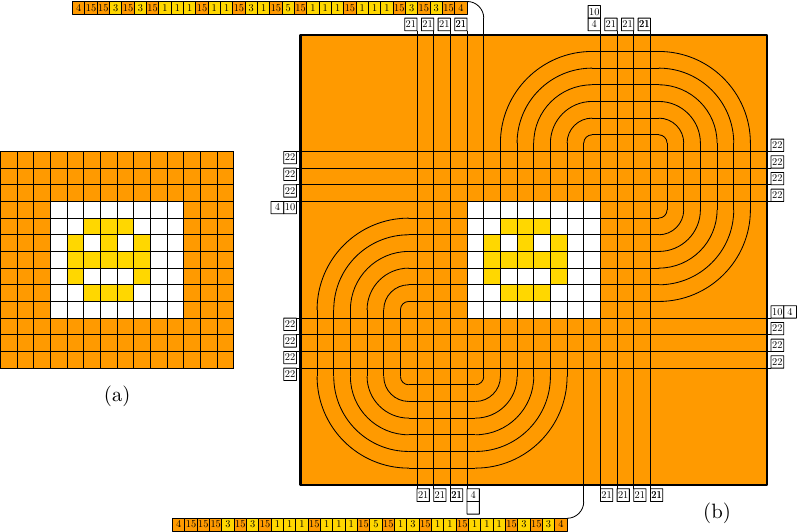}
		\caption {
			(a) An even more padded version of the nonogram from Figure~\ref {fig:hard1}.
			(b) The final construction including $k=3$ additional rows and columns of filled cells on all sides. Original filled cells are yellow; padding cells are orange. In the vital description, numbers are orange if they are at least $2k$ and yellow otherwise.}
		\label{fig:hard2} 
	\end {figure}
    
		We let $k = 1 + \max \left( \lfloor \frac w2 \rfloor, \lfloor \frac h2 \rfloor \right)$;
		this value is chosen to ensure that $2k$ is more than either $w$ or $h$.
		We first construct another padded  $w + 2 + 2k$ by $h + 2 + 2k$ grid: the original grid with a single empty and $k$ full rows added on all sides. Refer to Figure~\ref {fig:hard2}(a).
        All filled/empty cells that are added by this procedure are called the filled/empty \emph{padding cells}.
        
		Then, we construct a curved nonogram $C$ which consists of this grid surrounded by some additional potentially non-rectangular cells (which will be called \emph{boundary cells}).
        In total, it consists of $4k + 5$ curves: $k+1$ straight lines on each side of the input picture, plus \desc{1} very long curve $\ell$ which contains all original grid lines.
		Refer to Figure~\ref {fig:hard2}(b).
        Note that by construction all descriptions which consists of a single clue require their entire sequence to be filled.
        Settling all cells of these sequences to be filled also uniquely determines a fix for all sequences with descriptions consisting of two clues and we state the following observation.

        \begin{observation}\label{obs:other_descriptions}
            Any sequence other than $S^r_\ell$ has a description of length one or two.
            Moreover all boundary and filled padding cells can trivially be settled to be filled and all empty padding cells can trivially be settled to be empty.
            Every unsettled cell is an original cell and contained in $S^r_\ell$.
        \end{observation}

        Next we consider the vital description $D$.
        Note that we can partition $D$ into $2(w+h+4)+1$ (possibly empty) parts, s.t., these parts alternatingly correspond to the description of a column or row of $N$ and clues which require $4k-1$ consecutive filled cells (with the exception of the first and last part, which require exactly $k+1$ filled cells).
        We call the parts requiring $4k-1$ cells \emph{blockers}.
        Since the first and last cell in the vital sequence are filled, the first and last clue of the vital description are necessarily already fulfilled.
        Note two things.
        First there is a matching of already settled cells along the vital sequence and blockers, s.t., all blockers are fulfilled and second there are either $w+2$ or $h+2$ unsettled cells between two consecutive chains of $4k-1$ already filled cells and therefore no blocker can be fulfilled in such a space.
        Therefore this matching is the only possible realization of the blockers and we know that every chain of unsettled original cells in $C$ has to accommodate exactly the clues that the original row or column in $N$ had to realize.
        With this we state the following observation.

        \begin{observation}\label{obs:solution_equivalence}
            If we restrict the solution of $C$ to its original cells, we obtain exactly the solution of $N$.
        \end{observation}


        \begin{lemma}\label{lem:C_is_simple}
            If $N$ has a unique solution the constructed curved expert nonogram $C$ also has only a single solution.
            Moreover $C$ is simple.
        \end{lemma}
        \begin{proof}
            The first part of the lemma follows as a direct consequence of Observations~\ref{obs:other_descriptions} and~\ref{obs:solution_equivalence}. 
            Over all solutions of $C$, any padding and boundary cell in $C$ can have exactly one value (filled or unfilled) and any original cell can have at most as many values as the corresponding cell in $N$ over all valid solutions of $N$.
            If this solution of $N$ is unique, every original cell can have only one such value.

            The second part of the lemma statement is a consequence of Observation~\ref{obs:other_descriptions}.
            Since the value of all cells, which are not part of the vital sequence can trivially be settled, if we would apply the $\settle$ procedure to the vital sequence, we would settle all remaining cells, since they can have only one value in a solution (because this solution is unique).
        \end{proof}

        \subsection{Correctness}
        We are now ready to prove 
        the main theorem.

                \begin{restatable}{theorem}{hardness}\label{thm:problem_equivalence}
Solving \textsc{Simple Curved Expert Nonogram} is (a) at least as hard as \UNP{} and (b) in \textsf{NP}.
\end{restatable}
        \begin{proof}
            To prove statement (a) it suffices to show that we can construct $C$ based on a given nonogram $N$ in polynomial time and given a solution of \textsc{Simple Curved Expert Nonogram}, i.e., a filled version of $C$, we can construct a solution to \UNP{} for the instance $N$ in polynomial time.
            The first part is immediate as the construction as described in Section~\ref{sec:hardness_construction} which yields $C$ based on a given $N$ adds only a polynomially many cells to $N$ and the curves can be obtained by connecting at most a polynomial number of grid lines.
            Since $N$ has a unique solution by definition of \UNP{}, it follows from Lemma~\ref{lem:C_is_simple}.
                        
            The second part, i.e., constructing a solution for $N$ based on a given solution for $C$ follows from Observation~\ref{obs:solution_equivalence}.
            By simply settling all cells in $N$ according to the value of the original cells in $C$, we obtain the solution.

            To prove statement (b) it suffices to observe that, given a solution for \textsc{Simple Curved Expert Nonogram}, we can enumerate all polynomially many sequences, and check in polynomial time if their fix is consistent with their description.
            This concludes the proof.
        \end{proof}

        As previously mentioned we only show with Theorem~\ref{thm:problem_equivalence} that \textsc{Simple Curved Expert Nonogram} is at least as hard as \UNP{}.
        While it seems reasonable to expect \UNP{} to be \textsf{NP}-hard (equivalent to the generalized problem, i.e., finding a solution to any classic nonogram), the exact complexity of \UNP{} remains an open question.
        
    \change{We 
    note that the construction of our hardness proof 
    produces a simple curve arrangement, i.e., there are no three curves intersection in the same point, no two curves touch 
    without crossing, and no curves locally overlap in more than a single point.}

    \maarten {Thinking about it more, I wonder if the proper definition of "simple" should be "if the puzzle has a unique solution, then it can be found by iteratively applying settle on a single description". This maybe more closely matches the intuition behind a simple nonogram (i.e., if you can't make progress then it could be because the puzzle has no solution, or it has multiple solutions, or it has a unique solution but is not simple, but you don't really care) and then we could have cleaner theorem statements, at the cost of more subtle definitions. Anyway, probbaly too close to the deadline to get philosophical now.}

	\section {Conclusions}
	\label{sec:conclusion}
    
	We have shown that the concept of {\em simple} nonograms extends to curved nonograms to some extent.
	In general, simple curved nonograms are not necessarily easy to solve: even the problem of testing for progress on a single description is already as hard as solving a classic nonogram under the assumption that it 
    has a unique solution \change{and while we suspect this problem to be \NP-hard, the complexity of solving a simple curved expert nonogram remains an open question.}
	However, for the restricted classes of {\em basic} and {\em advanced} curved nonograms, we show that simple puzzles can be solved in polynomial time.
    It would be of interest how other measures of difficulty like the ones proposed by Batenburg and Kosters~\cite{DBLP:journals/icga/BatenburgK12} extend to curved nonograms.
	
	\bibliographystyle{abbrvurl}
	\bibliography{impop}

\appendix

\end{document}